\documentclass[conference]{IEEEtran}
\IEEEoverridecommandlockouts
\usepackage[T1]{fontenc}
\usepackage{cite}
\usepackage{amsmath,amssymb,amsfonts,amsthm}
\usepackage{url}

\usepackage{algorithm}
\usepackage{algpseudocode}
\usepackage{graphicx}
\usepackage{textcomp}
\usepackage{xcolor}
\usepackage{booktabs}
\usepackage{microtype}
\usepackage{tikz}
\usetikzlibrary{positioning,arrows.meta}
\usepackage{pgfplots}
\pgfplotsset{compat=1.18}

\newtheorem{theorem}{Theorem}
\newtheorem{lemma}{Lemma}
\newtheorem{remark}{Remark}

\newcommand{\sys}{FP-Sketch}

\usepackage{etoolbox}
\makeatletter
\patchcmd{\abstract}{---}{.\ }{}{}
\patchcmd{\IEEEkeywords}{---}{.\ }{}{}
\makeatother

\definecolor{seriesA}{HTML}{2A78D6}
\definecolor{seriesB}{HTML}{EB6834}
\definecolor{inkMuted}{HTML}{898781}
\definecolor{gridLine}{HTML}{E1E0D9}

\begin{document}

\title{Sketching the Error, Not the Product: Post Hoc Fault Recovery for Half Precision GPU Matrix Multiplication}

\author{
\IEEEauthorblockN{Pranav Napolean}
\IEEEauthorblockA{\textit{Department of Computer Science and Engineering} \\
\textit{National Institute of Technology Warangal}\\
Warangal, India \\
napoleanpranav@gmail.com}
\and
\IEEEauthorblockN{Vikas Srivastava}
\IEEEauthorblockA{\textit{Department of Mathematics} \\
\textit{National Institute of Technology Warangal}\\
Warangal, India \\
vsv@nitw.ac.in}
\and
\IEEEauthorblockN{Napolean Periathambi}
\IEEEauthorblockA{\textit{Executive Director} \\
\textit{AthenaHealth}\\
ppnapolean@gmail.com}
}

\maketitle

\begin{abstract}
Silent data corruption (SDC) from defective accelerators now interrupts large
scale training, yet deployed mitigations act on whole nodes. Algorithm based
fault tolerance (ABFT) for a single GEMM has to be fused into the kernel or
encode the operands, and it localizes at most one error per checksum. We
present \sys{}, a verifier that runs after an unmodified tensor core GEMM whose
half precision operands are accumulated and delivered at FP32. A sum sketch
detects corruption on every call. Hashed first moment sketches, confirmed by
independent recomputation, then localize several corrupted entries with no
false positives by construction, and each fault yields a coordinate and a
magnitude for fleet diagnosis. In floating point, sketch noise rather than
bucket collisions limits localization. We measure that noise and find that its
constant depends on the BLAS and the operand format and that the bucket count
must grow as $n^{2.57}$ for a square product. Sizing the bucket count by
measured noise rather than by a fitted power of $n$ raises recovery on eight
transformer shapes from $0.402$ to $1.000$, and measuring the noise at run time
adapts the bucket count to the kernel and the model.
Instruction level injection with NVBit shows that upsets in a live accumulator
are often only $2$ to $9\%$ of a typical entry, a population that output side
injection cannot produce. Output side injection recovers every fault, while
under NVBit the same engine sized for faults of typical magnitude recovers
$0.550$, and sizing for the measured magnitudes restores $1.000$. On
Llama-2-7B, guarding the MLP down projections removes $99.4\%$ (BF16) and
$99.9\%$ (FP16) of the perplexity damage caused by $2048$ bit flips, and the
clean path probe costs $0.78$ to $3.06$\,ms against GEMMs of $0.35$ to
$12.47$\,ms.
\end{abstract}


\section{Introduction}

Silent data corruption is a compute fault that returns a valid but wrong value
with no exception, NaN or crash. It has become a first order reliability
problem for large scale AI~\cite{ocp_sdc_ai_2025,sdc_survey_2026}. Training
multiplies the exposure. A single run occupies tens of thousands of
accelerators for weeks of synchronous work, and since every step builds on the
state the previous one produced, one wrong value is carried into every later
step~\cite{sdc_llm_2025,sdc_llm_train_2026}. Reported incidents understate the
damage. Instruction level injection into the matrix multiply instructions of
BF16 LLaMA training shows exponent bit upsets raising evaluation
loss~\cite{sdc_llm_train_2026}, and permanent GPU faults can distort parameters
with no NaN, no Inf, and no loss anomaly~\cite{llmprism2026}. The most damaging
corruption is thus the kind that is never counted.

Production mitigations work at the level of a node: they detect a misbehaving
node, evict it, and roll back~\cite{sdc_llm_2025}. They do not say which
product was wrong, or where. Below the node, ABFT embeds checksums in the
matrix product~\cite{huang1984,jou1984,gunnels2001}. GPU instances fuse the
checksum into a custom GEMM kernel~\cite{ding2011,wu2023ftgemm} or encode the
operands~\cite{grid_ecc_2025,vabft2026}, and the position weighted form
localizes one error per checksum~\cite{jou1984,vabft2026}. Mixed precision
training makes all of this hard to deploy. The GEMMs that dominate training run
in closed vendor libraries that a fault tolerance layer cannot modify, and a
checksum over output stored in BF16 is swamped by storage rounding. ABFT forced
to BF16 raised false positives on healthy nodes~\cite{sdc_llm_2025}. V-ABFT
shows that verifying inside a fused kernel, before quantization, restores FP32
level sensitivity~\cite{vabft2026}. Our question is how to obtain that
sensitivity, together with entry level localization of several simultaneous
errors, without owning the kernel.

\sys{} verifies the product after the fact. The unmodified GEMM delivers its
FP32 accumulator instead of a BF16 store, the verifier reads $A$, $B$ and $C$,
and the caller narrows $C$ afterwards. A sum sketch of the error $E = AB - C$,
computed without forming $AB$, decides on every call whether $C$ is wrong.
When it fires, first moment sketches weighted by row and column index recover
the coordinates of each isolated error as a ratio. This is the weighted
checksum device~\cite{jou1984}, applied inside hash buckets so that several
errors separate. A candidate is accepted only if an independent recomputation
of that entry disagrees with $C$.

Two effects that do not exist over exact rings decide whether this works on
real hardware. First, the ratio must resolve an index among thousands of rows,
so the rounding noise of the sketch, not bucket collisions, sets the bucket
count, and that noise depends on the summation order of the library. Second, a
fault in a live accumulator is buried by the additions that follow it, and
hardware faults are often far smaller than the flipped bit suggests. The
standard evaluation method, writing a flipped value into the finished product,
never produces them.

\medskip\noindent\textbf{Contributions.}
\begin{itemize}
\item A design for ex post facto verification of half precision tensor core
  GEMMs (\S\ref{sec:design}). A cheap $S$ only probe runs on every call, while
  localization and repair run only on the rare dirty call. Localization
  combines FP32 moment sketches with scaled index weights, peeling across hash
  rounds, a bucket count and neighbourhood search sized from the noise measured
  on each product, and an exact scan for nonfinite entries
  (Algorithms~\ref{alg:localize} and~\ref{alg:pipeline}). False positives are
  impossible under a trusted recompute assumption (\S\ref{sec:analysis}).
\item A measured noise law for hashed moment localization (\S\ref{sec:noise}).
  On cuBLAS it bounds all $65$ H100 measurements, makes the bucket count scale
  as $n^{2.57}$, and assigns FP16 operands $1.76$ times the BF16 constant. Its
  $n_2$ exponent differs between libraries. Sizing by measured noise rather
  than by a fitted power of $n$ raises recovery on eight transformer shapes from
  $0.402$ to $1.000$.
\item A fault magnitude result (\S\ref{sec:rq2}). Measured against NVBit
  accumulator injection, output side injection overstates localization
  recovery. Sized for the measured magnitudes, the engine recovers $1.000$
  (held out campaigns: $0.975$ with BF16 operands, $1.000$ with FP16), and the
  flipped bit position cannot be recovered ($0/80$).
\item An H100 evaluation (\S\ref{sec:eval}) with instruction level injection
  across single bit, multiple bit, adjacent burst, stuck at and whole word
  upsets, real models (including accumulator faults on a real Llama-2-7B layer,
  recovered at $0.975$ to $1.000$), per shape deployable costs, and baselines
  on identical faults.
\end{itemize}

\section{Background, Threat Model, and Related Work}
\label{sec:background}

\subsection{Mixed Precision GEMM and Where to Verify}
\label{sec:delivered}

For $A\in\mathbb{R}^{n_1\times n_2}$ and $B\in\mathbb{R}^{n_2\times n_3}$, a
tensor core GEMM loads FP16 or BF16 operands, accumulates each partial sum in
an FP32 register, and by default stores $C$ back at the operand precision. The
instrumented instruction stream confirms this: FP16 and BF16 operands both emit
an \texttt{HMMA} whose destination is an FP32 partial sum
(\texttt{HMMA.16816.F32}, \texttt{HMMA.*.F32.BF16}).

The store discards sixteen mantissa bits, and reading them back returns zeros,
not data. A verifier shown the stored value sees a number from which the
evidence has already been removed, whatever its threshold. We therefore fix the
deployment point: \emph{the product is verified at accumulation precision,
before the store narrows it}. The GEMM is asked for an FP32 output (cuBLASLt
compute type \texttt{32F} with an FP32 result, or a kernel that stores the
accumulator it already holds), and the caller narrows $C$ after verification.
What this costs is memory traffic. Delivering $C$ at FP32 and narrowing it
afterwards adds $6n_1n_3$ bytes to a product of $2n_1n_2n_3$ flops, a fraction
$3\pi/(\beta n_2)$ of the GEMM for achieved tensor core throughput $\pi$ and
bandwidth $\beta$. On an H100 that is about $443/n_2$: $11\%$ at $n_2=4096$,
$3.1\%$ at $14336$, and $58\%$ at $768$. V-ABFT makes the same observation for
a checksum fused into the kernel~\cite{vabft2026}. Here the same point is
reached without modifying the kernel.

\subsection{Threat Model}
\label{sec:threat}

\emph{In scope} are compute faults in a processing element, FMA, or accumulator
that leave the stored operands intact and produce wrong entries in $C$, with
$\|E\|_0 = s \ll n_1n_3$. The evaluation covers single and multiple bit upsets,
stuck at lines, whole word corruption, and nonfinite results.

\emph{Out of scope} is corruption already present in $A$ or $B$, since a
recomputation reads the same wrong operand and ECC and checkpoint validation
apply there. Faults outside the matrix product are also out of scope, as are
products that are never materialized. Fused attention kernels never store
$QK^\top$, so it cannot be verified after the fact. Within a training step, an
undetected fault in the output of layer $k$ becomes the input of layer $k{+}1$.
Checking every guarded GEMM stops the fault before it is laundered into
legitimate data.

\emph{Trusted recomputation.} The guarantee of zero false positives
(Lemma~\ref{lem:validation}) requires that the recomputation validating a
candidate is not corrupted by the same fault. Our implementation recomputes at
FP32 on the same device, so a permanent fault that corrupts both the product
and its recomputation identically is not covered.

\subsection{Related Work}
\label{sec:related}

\emph{ABFT.} Huang and Abraham introduced checksum based ABFT for matrix
operations~\cite{huang1984}, and the position weighted checksum of Jou and
Abraham recovers the index of a single error from a ratio of
syndromes~\cite{jou1984}. Gunnels et al.\ brought detection and rollback to
high performance matrix multiplication~\cite{gunnels2001}, and Ding et al.\
located and corrected errors online on GPUs~\cite{ding2011}. FT-GEMM fuses ABFT
into a custom GPU SGEMM kernel at $8.89\%$ average overhead over
cuBLAS~\cite{wu2023ftgemm}. A-ABFT and V-ABFT derive tighter detection
thresholds~\cite{abft_gpu_2014,vabft2026}, ApproxABFT loosens them to cut
recomputation~\cite{approxabft2023}, ATTNChecker specializes to attention and
extreme values~\cite{attnchecker2025}, and grid like codes correct errors
confined to at most two rows and columns by encoding the
operands~\cite{grid_ecc_2025}. Table~\ref{tab:related} contrasts these
approaches.

\emph{Verifying and correcting products.} Freivalds' test decides whether
$C=AB$ in $O(n^2)$~\cite{freivalds}, and coding theoretic tests detect sparse
errors~\cite{bennett2024}. G\k{a}sieniec et al.\ correct $s$ wrong entries after
the fact over a ring by sketching the error with Pagh's compressed
multiplication~\cite{gasieniec2017,pagh2013}. Roche and Wu and Wang give exact
algorithms over fields and integers~\cite{pernet2018,wu2024}. Hashed
sketches~\cite{countmin_2005,achlioptas2003} and invertible Bloom lookup
tables~\cite{iblt_2011} supply the bucket structure. All of these assume exact
arithmetic, where the sizing problem driven by noise (\S\ref{sec:noise}) does
not arise.

\emph{SDC characterization and injection.} Field and injection studies
characterize SDC in LLM training and GPU
kernels~\cite{sdc_llm_2025,sdc_llm_train_2026,llmprism2026,tung2026,mpgemmfi_2023,chai2026}.
Error injection at a high level is known to misestimate
resilience~\cite{cho2013}, and single bit flips are a minority of gate level
GPU faults~\cite{tung2026}. Section~\ref{sec:rq2} measures one specific
consequence of this for localizers.
\emph{Floating point reproducibility.} Summation order, and with it rounding,
depends on the parallel decomposition of a library~\cite{ahrens2020}, and the
rounding noise of GPU matmul is structured~\cite{yashwanth2025}.
Section~\ref{sec:noise} quantifies what that noise does to localization.

\begin{table*}[t]
\centering
\caption{Fault tolerance for a single GEMM. \emph{Vendor GEMM} asks whether the
         method can protect an unmodified library call. n/s: not stated in the
         cited source.}
\label{tab:related}
\footnotesize
\setlength{\tabcolsep}{4pt}
\begin{tabular}{@{}l l l l l l@{}}
\toprule
Method & Redundancy & Operands & Errors localized & Vendor GEMM & Reported cost \\
\midrule
Checksum ABFT~\cite{huang1984} & row and column checksums & any & 1 & post hoc form & $O(n^2)$ \\
Weighted checksum~\cite{jou1984} & position weighted checksums & any & 1 per checksum & post hoc form & $O(n^2)$ \\
FT-GEMM~\cite{wu2023ftgemm} & fused into custom kernel & FP32 & n/s & no & $8.89\%$ over cuBLAS \\
V-ABFT~\cite{vabft2026} & encoded, fused for FP32 threshold & BF16 to FP64 & 1 per row & no & n/s \\
Grid like ECC~\cite{grid_ecc_2025} & encoded, enlarged product & real & $\le 2$ rows, columns & no & $1.24$ to $1.37\times$ latency \\
G\k{a}sieniec et al.~\cite{gasieniec2017} & hashed sketch & exact ring & $s$ & yes & $\tilde O(n^2+sn)$ \\
\textbf{This work} & hashed moment sketches & BF16/FP16 & $s$ (sized) & yes (FP32 output) & probe $0.78$ to $3.06$\,ms \\
\bottomrule
\end{tabular}
\end{table*}

\section{Design}
\label{sec:design}

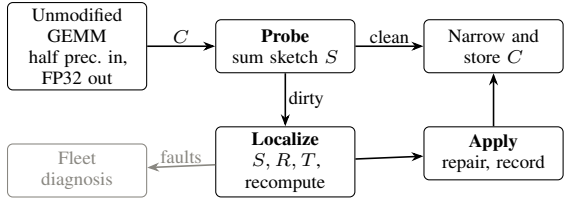
\begin{figure}[t]
\centering
\begin{tikzpicture}[font=\scriptsize,
  box/.style={draw, rounded corners=2pt, align=center, inner sep=2pt,
              minimum height=7mm, text width=17mm},
  arr/.style={-{Stealth[length=1.6mm]}, semithick},
  lab/.style={font=\scriptsize, inner sep=1pt}]
\node[box] (gemm) {Unmodified GEMM\\ half prec.\ in, FP32 out};
\node[box, right=9mm of gemm] (probe) {\textbf{Probe}\\ sum sketch $S$};
\node[box, right=9mm of probe] (store) {Narrow and\\ store $C$};
\node[box, below=7mm of probe] (loc) {\textbf{Localize}\\ $S,R,T$, recompute};
\node[box, below=7mm of store] (apply) {\textbf{Apply}\\ repair, record};
\node[box, below=7mm of gemm, text=inkMuted, draw=inkMuted] (fleet) {Fleet\\ diagnosis};
\draw[arr] (gemm) -- node[lab, above]{$C$} (probe);
\draw[arr] (probe) -- node[lab, above]{clean} (store);
\draw[arr] (probe) -- node[lab, right]{dirty} (loc);
\draw[arr] (loc) -- (apply);
\draw[arr] (apply) -- (store);
\draw[arr, inkMuted] (loc) -- node[lab, above, text=inkMuted]{faults} (fleet);
\end{tikzpicture}
\caption{\sys{} runs after the GEMM and before the store. The probe runs on
every call, and localization and repair run only when it fires.}
\label{fig:overview}
\end{figure}

\subsection{Overview}

Figure~\ref{fig:overview} shows the three stages. The \emph{probe} answers
whether $C$ is wrong and runs on every call. \emph{Localize} answers where and
by how much, and it does not modify $C$. \emph{Apply} writes the repairs and
emits one record per fault. We keep the stages as separate calls because their
costs differ by orders of magnitude (\S\ref{sec:rq4}). A deployment that wants
only a correct product can recompute on detection and still localize for
telemetry.

\subsection{Moment Sketches}
\label{sec:sketches}

Draw hashes $h_1:[n_1]\to[m]$ and $h_2:[n_3]\to[m]$ together with Rademacher
signs $v_1,v_2\in\{\pm1\}$~\cite{achlioptas2003}. Let $H_1\in\mathbb{R}^{m\times
n_1}$ hold $v_1(i)$ at $(h_1(i),i)$, and define $H_2$ in the same way. The
\emph{sum sketch}
\begin{equation}
  S = H_1 E H_2^\top = (H_1A)(BH_2^\top) - H_1CH_2^\top
  \label{eq:S}
\end{equation}
aggregates the error in each bucket and is formed without $AB$. Its $m\times
n_2$ and $n_2\times m$ factors come from scatter adds over $A$ and $B$, and
their product costs $m^2n_2$. The first moment sketches $R$ and $T$ repeat
\eqref{eq:S} with the rows of $A$ and $C$ weighted by $w_i = i/2^{\lceil\log_2
n_1\rceil}$ and the columns of $B$ and $C$ weighted by $w_j = j/2^{\lceil\log_2
n_3\rceil}$, with $i$ and $j$ counted from $1$. If bucket $(a,b)$ holds exactly
one error at $(i,j)$, then
\begin{equation}
  i = 2^{\lceil\log_2 n_1\rceil}R_{ab}/S_{ab},\quad
  j = 2^{\lceil\log_2 n_3\rceil}T_{ab}/S_{ab}.
  \label{eq:recover}
\end{equation}
The weights are scaled by a power of two so that no term of $R$ is larger than
the corresponding term of $S$. Without the scaling, an error with $|E_{ij}| >
\mathrm{fp32max}/n_1$ would overflow $R$ while $S$ stayed finite. Dividing by a
power of two only shifts an exponent, so every rounding in the accumulation is
unchanged.

\subsection{The Cheap $S$ Only Probe}
\label{sec:probe}

Production GEMMs are clean on the overwhelming majority of calls, so the stage
that decides deployed cost is the one that runs on a clean product. Deciding
\emph{whether} $C$ is wrong needs only $S$. If no bucket of $S$ exceeds the
probe threshold, no error large enough to localize is present, and $R$ and $T$
need not be formed. The probe therefore builds $S$ alone from a scatter over
the rows of $A$, a scatter over the rows of $C$ folded over its columns, and
one $m\times n_2\times m$ product with the factor $BH_2^\top$ of the $B$ side,
which is cached because weights stay the same across calls. It forms three of
the seven accumulated outputs that a localization round needs, reads $A$ once
and $C$ once, and recomputes no entry. Its output is a single bit, thresholded
as in \S\ref{sec:thresholds} together with an explicit finiteness test
(\S\ref{sec:nonfinite}). That bit can stay on the device until the caller reads
it once per step. The probe costs $O(n^2)$ (Theorem~\ref{thm:cost}) and never
walks the neighbourhood search. Its measured cost appears in \S\ref{sec:rq4}.

\subsection{Validation}
\label{sec:validation}

A bucket holding several errors yields a meaningless ratio, so every candidate
$(i,j)$ is validated. The inner product $\gamma = A_{i,*}B_{*,j}$ is recomputed
with FP32 accumulation, and the candidate is accepted only if $|\gamma -
C_{ij}| > \tau_{ij}$, where $\tau_{ij} = c\,\varepsilon\sum_k|A_{ik}||B_{kj}|$
bounds the rounding of that inner product~\cite{higham2002}, $c=100$, and
$\varepsilon = 2^{-23}$ is the FP32 machine epsilon. The accepted correction
is $\delta_{ij} = \gamma - C_{ij}$.

\subsection{Thresholds}
\label{sec:thresholds}

The analytic threshold
\[
\tau = 100\,n_2\,\varepsilon\,\max_{ik}|A_{ik}|\,\max_{kj}|B_{kj}|
\]
bounds the worst case accumulation noise in a bucket. Probe
and localizer threshold the same kind of sketch but run different tests. The
output of the probe is a single bit with no filter after it, so the threshold
must clear the \emph{maximum} over $m^2$ clean buckets:
$\mathit{thr} = \min(\tau,\ \mu\sqrt{2\ln m^2}\,\hat\sigma)$ with $\hat\sigma =
1.2533\cdot\mathrm{mean}(\min(|S|,5\,\mathrm{mean}|S|))$ and margin $\mu = 4$.
This clipped mean needs no sort, so the estimate stays on the device.

The localizer selects candidates, and a false candidate costs one
recomputation, so its threshold may be looser. It uses
$\tau_c = \min(\tau,\, k\,\hat\sigma_{\mathrm{MAD}})$ with $k = 2$, where
$\hat\sigma_{\mathrm{MAD}} = 1.4826\cdot\mathrm{median}\big(\big||S| -
\mathrm{median}|S|\big|\big)$ estimates the clean bucket noise from the median
absolute deviation of $|S|$. About $4.6\%$ of clean buckets clear $\tau_c$ by chance.
Validation discards all of them, but each one would still walk the whole
neighbourhood search. A lone fault of size $\delta$ decodes its row to within
about $e\max(n_1,n_3)\,\sigma/\delta$, where $e$ is the normalized index error
(median $0.58$, $95$th percentile $2.5$). A bucket can therefore localize within
radius $r$ only if its signal is at least of order $\max(n_1,n_3)\,\sigma/r$.
The localizer adds this as a floor,
$\tau_d = \max(n_1,n_3)\,\hat\sigma_{\mathrm{MAD}}/\max(r,\tfrac12)$, keeps the
buckets with $|S_{ab}| > \max(\tau_c,\tau_d)$, and retains at most
$K = \max(64, 8s)$ of them, ranked by $|S_{ab}|$.

\subsection{Sketch Noise Sets the Bucket Count}
\label{sec:noise}

Over a ring, collisions alone set $m$. The condition
$m \ge \max(3c\Delta, \sqrt{3cs})$ isolates an error with probability at least
$1-1/c$, where $\Delta$ bounds the errors per row or column
(Lemma~\ref{lem:isolation}). In floating point, \eqref{eq:S} is the difference
of two independently rounded accumulations, and its noise $\sigma$ grows with
the number of entries per bucket. On an H100 (cuBLAS, BF16 operands, FP32
accumulation), normalizing each measured $\sigma$ by
$2^{-24}\sqrt{n_1n_3n_2}\,\mathrm{rms}(C)/m$ gives $0.1726$ to $0.1749$ across
$n_3\in[1024,16384]$ and $0.1717$ to $0.1783$ across $m\in[32,512]$. Across
$n_2\in[1024,65536]$, however, the same normalized value climbs from $0.162$ to
$0.219$. The $n_2$ exponent is $0.57$ rather than $1/2$, and with it the
normalized value stays flat to within $1.10\times$ over that sweep. FP16
operands are noisier than BF16 at every one of $19$ identical shapes, by
$1.20$ to $1.76\times$ with a median of $1.47$. Fitting all $65$ H100
measurements ($43$ BF16, $22$ FP16) as an upper envelope gives
\begin{equation}
  \sigma \lesssim c_f\,2^{-24}\sqrt{n_1n_3n_2}\,(n_2/8192)^{0.07}\;\mathrm{rms}(C)/m
  \label{eq:sigma}
\end{equation}
with $c_f = 0.19$ for BF16 and $c_f = 0.33$ ($1.76\times0.19$) for FP16. No
measurement exceeds \eqref{eq:sigma}. The earlier fit $0.17\cdot2^{-24}
\sqrt{n_1n_3n_2}\,\mathrm{rms}(C)/m$, with no format term, was exceeded by
$17$ of the $43$ BF16 measurements, by up to $1.29\times$.

Rounding in \eqref{eq:recover} perturbs the recovered row by
$(\nu_R - i\,\nu_S)/\delta$, where $\nu_R$ and $\nu_S$ are the bucket noises and
$\delta = |E_{ij}|$. The scatter is $\Theta(\max(n_1,n_3)\,\sigma/\delta)$ with
zero bias, and we call $B = \max(n_1,n_3)\,\sigma/\delta$ the \emph{index
budget}. Keeping $B$ below a target $B^\star$ requires
\begin{equation}
  \begin{split}
  m_{\mathrm{num}} \gtrsim{}& \frac{\max(n_1,n_3)}{B^\star\,\delta}\,
  c_f\,2^{-24}\sqrt{n_1n_3n_2}\\
  &\cdot(n_2/8192)^{0.07}\,\mathrm{rms}(C)
  \end{split}
  \label{eq:mnum}
\end{equation}
which is $\Theta(n^{2.57})$ for a square product. \sys{} sizes for $\delta =
\mathrm{rms}(C)$ with $B^\star = 0.3$. With the collision requirement
$m_{\mathrm{comb}} = \lceil\max(3c\Delta, \sqrt{3cs})\rceil$ at $c = 2$, it
starts from
\begin{equation}
  \begin{split}
  m = \min\big(&\max(m_{\mathrm{comb}},\ m_{\mathrm{num}},\ 16),\\
  &\max(m_{\mathrm{comb}},\ m_{\max})\big)
  \end{split}
  \label{eq:mstart}
\end{equation}
The cap
$m_{\max} = \max(16, \lfloor\sqrt{n_1n_3/16}\rfloor)$ keeps an average of $16$
entries in each bucket. Below that population the noise estimates collapse and
clean products read as dirty, which we observed at $m = 966$ on a
$1024\times1024$ output.

The constant belongs to the BLAS, not to the algorithm. $\sigma$ measures the
disagreement between two summation orders, and the blocking of a library sets
how far it departs from sequential summation~\cite{ahrens2020}. The $n_2$
exponent is $+0.587$ for an explicitly sequential sum, $+0.200$ through a
blocked CPU matmul (oneDNN) and $+0.133$ for the sketch built on it, against
$+0.577$ for the sketch on cuBLAS. A rule fitted on one library and deployed on
the other underprovisions $m$ by $4.4\times$ at $n_2 = 8192$. For this reason
\sys{} uses \eqref{eq:sigma} only as a starting point and measures $\sigma$ at
run time (\S\ref{sec:radius}).

\subsection{Fault Magnitude, Calibration, and the Neighbourhood Search}
\label{sec:radius}

Equation~\eqref{eq:mnum} is written for a fault the size of a typical entry,
$\rho \equiv \delta/\mathrm{rms}(C) = 1$. Faults in a live accumulator can be
much smaller (\S\ref{sec:rq2}), which multiplies $B$ by $1/\rho$. A larger $m$
and a wider validation window both absorb that factor. Validating a
$(2r{+}1)^2$ neighbourhood of the rounded index costs that many inner products
per candidate, and on the GPU this dominates: at the Llama-2-7B
\texttt{down\_proj} shape ($4096{\times}11008{\times}4096$), localization with
radius $16$ took $1153$ to $1196$\,ms at every $m$ from $112$ to $960$, while
radius $4$ took $96$ to $97$\,ms. \sys{} therefore declares the smallest
magnitude it must localize, $\rho_{\min}$ (default $0.02$), sizes the bucket
count for a target radius $r^\star = 4$, and lets the radius grow toward
$r_{\max} = 16$ only when $m$ cannot grow further.

\emph{Calibration.} Equation~\eqref{eq:sigma} is fitted on one GPU, one library
and Gaussian operands. Other kernels and real activations change the noise. On
a Llama-2-7B \texttt{down\_proj} layer, products from the tiled kernel used for
injection carried about $8\times$ the noise of cuBLAS products on the same
operands. The plan for a dirty call is therefore computed from the noise
measured on that call. Let $\hat\sigma$ be the probe estimate of
\S\ref{sec:thresholds} if the probe has just run on this $C$, and otherwise the
estimate $\hat\sigma_{\mathrm{MAD}}$ of a freshly built $S$. Let
$\widehat{\mathrm{rms}}(C)$ be the rms of the finite entries of $C$ taken at a
stride of $\max(1,\lfloor n_1n_3/65536\rfloor)$, with magnitudes clipped at
$64$ times their median. Then
\begin{gather*}
B = \frac{\max(n_1,n_3)\,\hat\sigma}{\rho_{\min}\,\widehat{\mathrm{rms}}(C)},\\
r_B = \begin{cases} 0 & B \le 0.05,\\ \max(2,\lceil 2.5\,B\rceil) & \text{otherwise.}\end{cases}
\end{gather*}
The factor $2.5$ covers the $95$th percentile of measured index error. If
$r_B \le r^\star$, localization keeps $m_{\mathrm{loc}} = m$ and uses $r = r_B$.
Otherwise it builds its sketches at
\[
m_{\mathrm{loc}} = \max\big(m,\ \min(16\lceil\lceil m\,r_B/r^\star\rceil/16\rceil,\
m_{\max},\ m_{\mathrm{mem}})\big),
\]
where $m_{\mathrm{mem}}$ is the largest $m$
with $4(3m^2+6mn_2)$ bytes within $2$\,GiB. The radius at that size follows from
$B' = B\,m/m_{\mathrm{loc}}$ as $r = 0$ if $B' \le 0.05$ and
$r = \min(r_{\max}, \max(2,\lceil 2.5\,B'\rceil))$ otherwise. If $\hat\sigma$
or $\widehat{\mathrm{rms}}(C)$ is not a positive finite number, the plan falls
back to the same rule with $\sigma$ taken from \eqref{eq:sigma}.

\emph{Recalibrating the probe.} The probe keeps its own $m$, which sets its
detection floor. It averages $\hat\sigma$ over the calls it reports clean,
without a host synchronization, and after $16$ calls and then every $1024$ it
computes $B_1 = \max(n_1,n_3)\,\bar\sigma/\widehat{\mathrm{rms}}(C)$ and
\[
m' = \max\big(m_{\mathrm{comb}},\ 16,\ \min(\lceil m\,B_1/B^\star\rceil,\ m_{\max})\big).
\]
When $m'$ lies outside $[m/1.25,\ 1.25\,m]$, the engine rebuilds
its hashes, buffers, caches and graphs at $m'$. The formula therefore matters
only until the first measurement. After repair, \sys{} probes the product once
more with a fresh hash round and recomputes it if the probe still fires, so a
missed fault that the probe can still see does not reach the output. None of
these steps runs on a clean product.

\subsection{Peeling Across Hash Rounds}
\label{sec:peel}

A fault occupies exactly one bucket per hash round, so subtracting a recovered
fault inside the round in which it was found unlocks nothing. Across rounds it
does. \sys{} draws $k = 3$ independent hash rounds and, before selecting
candidates in a round, subtracts every finite fault confirmed so far from the
$S$, $R$ and $T$ of that round. A bucket that would still hold a mixture then
becomes a singleton. Only validated faults are subtracted, and every new
candidate is validated again, so peeling cannot create a false correction.

\subsection{Nonfinite and Out of Range Entries}
\label{sec:nonfinite}

An infinite or NaN entry is contagious in a linear sketch. One inf in the row
scatter poisons a bucket row, and every fault sharing that row becomes
unrecoverable. \sys{} finds such entries, and any entry with
$|C_{ij}| > \mathrm{fp32max}/16$, with an exact elementwise scan and keeps at
most $K$ of them. Each is recomputed as an absolute value rather than a delta,
because $\mathrm{NaN} + \delta$ is NaN, and $C$ holds the recomputed value while
the sketches are built. These entries are not subtracted when peeling, since
their error is already absent from the sketch. The probe tests the finiteness
of $S$ explicitly: $(|S| > \mathit{thr})$ is false for a NaN bucket, so a probe
that only compared values would report such a product clean.

\begin{algorithm}[t]
\caption{One round of floating point localization}
\label{alg:localize}
\begin{algorithmic}[1]
\Require $A$, $B$, $C$ (FP32), round $(h_1,v_1,h_2,v_2)$ at $m_{\mathrm{loc}}$, radius $r$, confirmed set $F$ with absolute entries $F_{\mathrm{abs}}\subseteq F$
\State $S,R,T \gets$ sketches of $E = AB-C$ with scaled weights (\S\ref{sec:sketches})
\ForAll{$(i,j,\delta)\in F\setminus F_{\mathrm{abs}}$} \Comment{peel across rounds}
  \State subtract $v_1(i)v_2(j)\,\delta\cdot(1,w_i,w_j)$ from $(S,R,T)$ at $(h_1(i),h_2(j))$
\EndFor
\State $\hat\sigma \gets$ MAD estimate of $S$
\State $\tau_c \gets \min(\tau, 2\hat\sigma)$, \ $\tau_d \gets \max(n_1,n_3)\,\hat\sigma/\max(r,\tfrac12)$
\State $Q \gets$ the $K$ largest $|S_{ab}|$ with $|S_{ab}| > \max(\tau_c,\tau_d)$
\State $U \gets \{(a,b)\in Q : S_{ab}\ne0$ and $S_{ab},R_{ab},T_{ab}$ finite$\}$
\ForAll{$(a,b)\in U$}
  \State $(\hat\imath_{ab},\hat\jmath_{ab}) \gets$ \eqref{eq:recover}, rounded to integers
\EndFor
\ForAll{offsets $(d_i,d_j)$ with $|d_i|,|d_j|\le r$, by Chebyshev then Manhattan distance}
  \ForAll{$(a,b)\in U$, in a fixed order}
    \State $(i,j) \gets (\hat\imath_{ab}+d_i,\hat\jmath_{ab}+d_j)$
    \If{$(i,j)$ lies outside $C$ \textbf{or} $(i,j)\in F$} \textbf{continue} \EndIf
    \State $\gamma \gets A_{i,*}B_{*,j}$ \Comment{FP32 accumulation}
    \If{$|\gamma - C_{ij}| > \tau_{ij}$}
      \State $F \gets F\cup\{(i,j,\gamma - C_{ij})\}$, \ $U \gets U\setminus\{(a,b)\}$
    \EndIf
  \EndFor
\EndFor
\State \Return $F$
\end{algorithmic}
\end{algorithm}

\begin{algorithm}[t]
\caption{Verify, localize, and apply}
\label{alg:pipeline}
\begin{algorithmic}[1]
\Require $A$, $B$, $C$, probe size $m$, rounds $k$, $\rho_{\min}$, $r^\star$, $r_{\max}$, $K$
\Statex \textbf{$S$ only probe} (every GEMM, \S\ref{sec:probe})
\State reject $C$ unless it is FP32 \Comment{\S\ref{sec:delivered}}
\State $S \gets$ sum sketch \eqref{eq:S} at $m$ \Comment{$B$ side cached}
\State $\hat\sigma \gets 1.2533\cdot\mathrm{mean}(\min(|S|,5\,\mathrm{mean}|S|))$
\State $\mathit{thr} \gets \min(\tau, \mu\sqrt{2\ln m^2}\,\hat\sigma)$
\If{$\max|S| \le \mathit{thr}$ \textbf{and} $S$ finite}
  \State add $\hat\sigma$ to the recalibration average (\S\ref{sec:radius})
  \State \Return \textsc{clean}
\EndIf
\Statex \textbf{Localize} (dirty GEMM only)
\State $F_{\mathrm{abs}} \gets \{(i,j,A_{i,*}B_{*,j}) : C_{ij}$ nonfinite or $|C_{ij}|>\mathrm{fp32max}/16\}$, at most $K$
\State $F \gets F_{\mathrm{abs}}$, and set those $C_{ij}$ to their recomputed values
\State $(m_{\mathrm{loc}}, r) \gets$ plan from $\hat\sigma$ and $\widehat{\mathrm{rms}}(C)$ \Comment{\S\ref{sec:radius}}
\For{$t = 1,\dots,k$}
  \State $F \gets$ Algorithm~\ref{alg:localize} with a fresh round at $m_{\mathrm{loc}}$, $r$, $F$
\EndFor
\State restore the entries of $F_{\mathrm{abs}}$ in $C$
\Statex \textbf{Apply}
\ForAll{$(i,j,\delta)\in F$}
  \State $C_{ij} \gets \delta$ if $(i,j,\delta)\in F_{\mathrm{abs}}$, else $C_{ij} \gets C_{ij}+\delta$
  \State emit record $(i,j,\delta)$
\EndFor
\If{the probe fires on $C$ with a fresh round} $C \gets AB$ at FP32 \EndIf
\end{algorithmic}
\end{algorithm}

\section{Analysis}
\label{sec:analysis}

\begin{lemma}[Validation soundness]
\label{lem:validation}
If the recomputation of $\gamma$ is not affected by the fault and $\tau_{ij}$
bounds its FP32 rounding, every accepted $(i,j)$ is a corrupted entry.
\end{lemma}
\begin{proof}
By assumption $|\gamma - (AB)_{ij}| \le \tau_{ij}$, so $|\gamma - C_{ij}| >
\tau_{ij}$ implies $C_{ij} \ne (AB)_{ij}$.
\end{proof}

\begin{lemma}[Collision]
\label{lem:collision}
For a fixed error $(i,j)$ under independent uniform hashes,
$\Pr[(i,j)\text{ not isolated}] \le 2(\Delta-1)/m + s/m^2$.
\end{lemma}
\begin{proof}
An error in the same row or column collides with probability $1/m$. Any other
error collides only if both hashes agree, which happens with probability
$1/m^2$. A union bound completes the argument.
\end{proof}

\begin{lemma}[Isolation per round]
\label{lem:isolation}
If $m \ge 3c\Delta$ and $m^2 \ge 3cs$ with $c > 1$, each error is isolated in a
round with probability at least $1 - 1/c$.
\end{lemma}

\begin{theorem}[Floating point recovery]
\label{thm:float}
Let error $(i,j)$ be isolated in bucket $(a,b)$ with $|S_{ab}| >
\max(\tau_c,\tau_d)$ and among the $K$ largest such buckets, and let the rounded
index of~\eqref{eq:recover} lie within $r$ of $(i,j)$ in each coordinate. Then
Algorithm~\ref{alg:localize} returns $(i,j,\delta)$ with
$|\delta - E_{ij}| \le \tau_{ij}$ and returns no uncorrupted entry. Since the
index scatter is $\Theta(B)$ (\S\ref{sec:noise}), the index condition holds with
high probability once $r \ge 2.5B$.
\end{theorem}
\begin{proof}
The search reaches $(i,j)$, and validation accepts it because $|E_{ij}|$
exceeds the rounding bound. Lemma~\ref{lem:validation} excludes every clean
position visited.
\end{proof}

\begin{theorem}[Recovery over rounds]
\label{thm:rounds}
Under Lemma~\ref{lem:isolation} and the index condition of
Theorem~\ref{thm:float}, $k$ independent rounds recover all $s$ errors with
probability at least $1 - s\,c^{-k}$.
\end{theorem}
\begin{proof}
An error is missed in all rounds with probability at most $c^{-k}$, and a union
bound over the $s$ errors gives the result.
\end{proof}

\begin{remark}[Peeling]
Peeling across rounds subtracts only validated faults, so a bucket isolated for
an error in the independent process stays isolated, up to the rounding residual
of each subtracted $\delta$ and the candidate cap $K$. A proof that covers both
effects is open. We use Theorem~\ref{thm:rounds} as the bound, and the
evaluation measures the implemented decoder.
\end{remark}

\begin{theorem}[Cost]
\label{thm:cost}
With the sketch of the $B$ side cached, the probe performs
$O(n_1n_2 + n_1n_3 + m^2n_2)$ operations, and localization performs
$O\big(k(n_1n_2 + n_1n_3 + m_{\mathrm{loc}}^2n_2) + kK(2r{+}1)^2n_2\big)$,
against $O(n_1n_2n_3)$ for recomputing $AB$.
\end{theorem}

The repaired entry is a fresh FP32 inner product. It is accurate to the rounding
of that recomputation rather than equal to $(AB)_{ij}$.

\section{Implementation}
\label{sec:impl}

\sys{} is a PyTorch engine. The row and column scatters run as segmented
kernels without atomic operations, compiled at run time through NVRTC with
native BF16 and FP16 variants, so the $O(n^2)$ passes never materialize a
weighted temporary. The sketch of the $B$ side is cached and keyed on tensor
identity and version, so an optimizer update in place invalidates it.
Activations change on every call and are not cached. The probe can be replayed
from a CUDA graph. The graph writes its noise estimate to a static buffer that
recalibration reads, and it can accumulate the verdict into a flag on the
device, so a caller synchronizes once per step rather than once per GEMM.
\texttt{probe}, \texttt{localize} and \texttt{apply\_corrections} are separate
calls, and a narrowed $C$ is rejected rather than tolerated. TF32 must be
disabled for the contractions of the verifier itself, because enabling it
raised the measured sketch noise floor $419\times$. Code and raw logs will be released.

\section{Evaluation}
\label{sec:eval}

\subsection{Methodology}
\label{sec:setup}

All measurements use an NVIDIA H100 80\,GB (driver 580.95.05, CUDA 13.0,
PyTorch 2.13) with TF32 disabled. Operands are BF16 or FP16 as stated, and
products are accumulated and delivered at FP32. \emph{Output side} faults flip
IEEE 754 bit 26 ($256\times$) of an entry of the finished $C$.
\emph{Instruction level} faults use NVBit~\cite{nvbit2019}. A tool instruments
every \texttt{HMMA} in a rectangular BF16 WMMA GEMM with $16^3$ tiles and
corrupts a live FP32 accumulator register partway through the multiply. A flip
of bit 26 on an accumulator holding $16$ produced exactly $4096$. We use a
custom NVBit tool rather than NVBitFI~\cite{tsai2021nvbitfi}, which was released
for NVBit 1.5.5 and CUDA 11.2 before the Hopper architecture, so that we can
target the FP32 accumulator of \texttt{HMMA} instructions with up to four
simultaneous sites, arbitrary bit masks, and stuck at modes. Ground truth comes
from differencing against an uninstrumented run, not from the report of the
injector. A trial whose difference is nonfinite has no defined ground truth and
is not scored. Recovery is the fraction of injected faults located and
repaired, reported with $95\%$ Wilson intervals.

Every result uses the calibrated engine of \S\ref{sec:radius} with $k=3$,
$\rho_{\min}=0.02$, $r^\star = 4$, $r_{\max}=16$, and
$\tau_c = \min(\tau, 2\hat\sigma_{\mathrm{MAD}})$, unless stated otherwise. The
only exception is the fitted baseline of Table~\ref{tab:rect}, which is
described with it. Figure~\ref{fig:cost} and the timings of the dirty path in
\S\ref{sec:rq4} were measured with the GPU to themselves. The timings reported
with Tables~\ref{tab:realmodel} and~\ref{tab:baselines} were measured while
NVBit campaigns shared the same GPU.

\subsection{RQ1: Accuracy on Deployable Shapes}
\label{sec:rq1}

\begin{table}[t]
\centering
\caption{Recovery on transformer layer shapes, one output side fault per trial,
         60 trials per shape and format. \emph{Fitted}: BF16, bucket count from
         a rule fitted as $n^{3/2}$, localization at that $m$ with exact
         rounding and no measured plan. \emph{Law}: the calibrated engine, with
         identical recovery in BF16 and FP16. $m$ is the probe's bucket count.
         FPR is $0$ in every cell.}
\label{tab:rect}
\footnotesize
\setlength{\tabcolsep}{2.4pt}
\begin{tabular}{@{}l r r r r r@{}}
\toprule
 & \multicolumn{2}{c}{Fitted} & \multicolumn{3}{c}{Law} \\
$n_1{\times}n_2{\times}n_3$ & $m$ & Rec. & $m$ BF16 & $m$ FP16 & Rec. \\
\midrule
$4096{\times}4096{\times}4096$    & 48  & 0.750 & 48  & 68   & 1.000 \\
$8192{\times}4096{\times}4096$    & 48  & 0.367 & 110 & 193  & 1.000 \\
$8192{\times}4096{\times}14336$   & 105 & 0.183 & 358 & 630  & 1.000 \\
$8192{\times}14336{\times}4096$   & 48  & 0.383 & 224 & 393  & 1.000 \\
$16384{\times}8192{\times}8192$   & 128 & 0.383 & 649 & 1142 & 1.000 \\
$32768{\times}4096{\times}4096$   & 363 & 0.333 & 874 & 1538 & 1.000 \\
$4096{\times}4096{\times}32768$   & 363 & 0.400 & 874 & 1538 & 1.000 \\
$4096{\times}32768{\times}4096$   & 48  & 0.417 & 127 & 223  & 1.000 \\
\midrule
\textbf{Pooled} & & \textbf{0.402} & & & \textbf{1.000} \\
 & & {\footnotesize [0.359, 0.447]} & & & {\footnotesize [0.992, 1.000]} \\
\bottomrule
\end{tabular}
\end{table}

Table~\ref{tab:rect} shows that sizing by measured noise is necessary at
deployable shapes. The fitted rule recovers $193/480$, and none of its $287$
misses is a detection failure. Every one was detected and then mislocalized,
which is the $\Theta(n)$ failure of the index budget described in
\S\ref{sec:noise}. The worst shape is $8192{\times}4096{\times}14336$ at
$0.183$. The calibrated engine recovers all $480$ faults in each format, so
$960$ of $960$ in total, with no false positives. The law also exposes a scaling
limit. At $65536^3$ it asks for $m = 48010$, which would need $25.8$\,GB for
$S$, $R$ and $T$, and the population cap stops $m$ at $16384$
(\S\ref{sec:discussion}).

\subsection{RQ2: Hardware Faults and Fault Magnitude}
\label{sec:rq2}

Instruction level injection produces a fault population that output side
injection cannot. NVBit corrupts the accumulator at an arbitrary point of the
$n_2$ loop, so an early hit perturbs a partial sum that is still small. The
resulting $|E_{ij}|$ span two orders of magnitude, and the smallest are $2$ to
$9\%$ of a typical entry. A value written into a finished $C$ is always $O(1)$
relative to that entry.

\begin{figure}[t]
\centering
\begin{tikzpicture}
\begin{axis}[
  width=\columnwidth, height=4.3cm,
  xmode=log, x dir=reverse,
  xtick={1,0.2,0.05,0.02}, xticklabels={1.00,0.20,0.05,0.02},
  xlabel={$\rho_{\min}$ (smallest fault localized, $\times\mathrm{rms}(C)$)},
  ylabel={Recovery},
  ymin=0.35, ymax=1.03,
  ytick={0.4,0.5,0.6,0.7,0.8,0.9,1.0},
  tick label style={font=\footnotesize}, label style={font=\footnotesize},
  axis line style={inkMuted}, tick style={inkMuted},
  ymajorgrids, grid style={gridLine},
  legend style={font=\footnotesize, draw=none, fill=none},
  legend pos=south west,
  clip=false]
\addplot[color=seriesA, line width=1pt, mark=*, mark size=2pt,
  error bars/.cd, y dir=both, y explicit, error bar style={seriesA, line width=0.6pt}]
  coordinates {
    (1.00,0.550) += (0,0.104) -= (0,0.109)
    (0.20,0.887) += (0,0.053) -= (0,0.087)
    (0.05,0.975) += (0,0.018) -= (0,0.062)
    (0.02,1.000) += (0,0.000) -= (0,0.046)
  };
\addplot[color=seriesB, line width=1pt, mark=square*, mark size=2pt,
  error bars/.cd, y dir=both, y explicit, error bar style={seriesB, line width=0.6pt}]
  coordinates {
    (1.08,0.487) += (0,0.108) -= (0,0.106)
    (0.216,0.812) += (0,0.071) -= (0,0.099)
    (0.054,0.938) += (0,0.035) -= (0,0.076)
    (0.0216,0.963) += (0,0.024) -= (0,0.068)
  };
\legend{BF16, FP16}
\end{axis}
\end{tikzpicture}
\caption{Recovery under NVBit accumulator injection with the calibrated engine,
$4096{\times}2048{\times}4096$, $80$ faults per point ($40$ trials, two sites
each), Wilson $95\%$ intervals. FP16 points are offset slightly along the
axis for legibility. There were no detection failures and no false positives.
Every miss was a localization failure except one FP16 fault at
$\rho_{\min}=1.00$, which lay below its rounding bound. The same engine
recovers every output side fault (Table~\ref{tab:rect}).}
\label{fig:rho}
\end{figure}
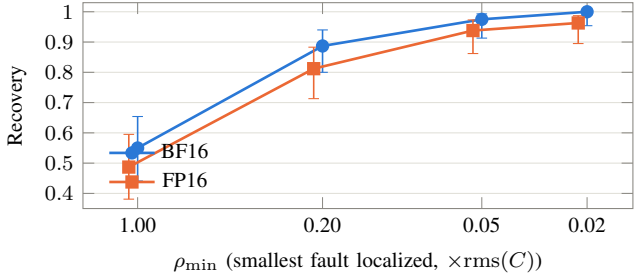

Figure~\ref{fig:rho} shows what follows. Declaring $\rho_{\min} = 1$, which
sizes the plan for faults of typical magnitude, gives a recovery of $0.550$
with BF16 operands and $0.487$ with FP16, although the same engine recovers
every output side fault in Table~\ref{tab:rect}. The plan sets the bucket
count, the radius and the floor $\tau_d$ for that magnitude, and the smaller
accumulator faults do not fit it. Declaring the measured magnitude,
$\rho_{\min} = 0.02$, raises recovery to $1.000$ and $0.963$. Recovery figures
obtained by writing into a completed product are therefore upper bounds for a
localizer. Because $\rho_{\min}$ was chosen with this campaign in mind, the held
out campaigns on fresh fault sites are the ones that validate it. They recover
$78$ of $80$ faults with BF16 operands ($0.975$, $[0.913, 0.993]$) and $80$ of
$80$ with FP16 ($1.000$, $[0.954, 1.000]$), with no false positives, and every
one of the $80$ products was delivered correct after repair and the second
probe.

\begin{table}[t]
\centering
\caption{Recovery under NVBit injection by fault pattern,
         $4096{\times}2048{\times}4096$, two sites per trial, $\rho_{\min}=0.02$.
         \emph{Above}: BF16 recovery among faults whose error exceeds the
         rounding bound $\tau_{ij}$ of the entry. Below that bound no verifier
         can separate a fault from rounding. $38$ to $40$ scored trials per
         cell. \emph{Stuck at 1}: one BF16 trial of $40$ and no FP16 trial
         produced an observable corruption.}
\label{tab:faultmodels}
\footnotesize
\setlength{\tabcolsep}{3pt}
\begin{tabular}{@{}l c c c c@{}}
\toprule
Fault pattern & BF16 & Above & FP16 & FPR \\
\midrule
1 bit, bits 26 and 27                  & 0.975  & 0.975 & 0.975  & 0 \\
1 bit, anywhere                        & 0.423  & 0.767 & 0.382  & 0 \\
2 random bits                          & 0.671  & 0.841 & 0.662  & 0 \\
3 random bits                          & 0.775  & 0.849 & 0.787  & 0 \\
4 random bits                          & 0.936  & 0.948 & 0.962  & 0 \\
2 adjacent bits~\cite{tsai2021nvbitfi} & 0.440  & 0.786 & 0.372  & 0 \\
3 adjacent bits                        & 0.380  & 0.682 & 0.449  & 0 \\
2 exponent bits                        & 0.987  & 0.987 & 0.987  & 0 \\
Stuck at 0                             & 1.000  & 1.000 & 1.000  & 0 \\
Stuck at 1                             & 1/1    & 1/1   & silent & 0 \\
Whole word replacement                 & 0.988  & 0.988 & 1.000  & 0 \\
\bottomrule
\end{tabular}
\end{table}

\emph{Upsets of several bits.} Repair does not depend on how many bits flipped.
A localized entry is replaced by a fresh recomputation, so a two bit, stuck at
or whole word corruption is repaired exactly as a single flip is. The number of
flipped bits matters only through the magnitude of the error it produces, and
each regime is handled by a different mechanism. Flips confined to low mantissa
bits can leave an error below the rounding bound of the entry, where no
verifier can separate the fault from rounding, and Table~\ref{tab:faultmodels}
reports recovery above the bound separately so that these are not counted as
losses. Flips that reach high exponent bits produce infinite, NaN or out of
range values, which the exact scan of \S\ref{sec:nonfinite} recovers without a
sketch. Everything in between is localized by Algorithm~\ref{alg:localize},
subject to the index budget of \S\ref{sec:radius}.

Table~\ref{tab:faultmodels} tests one to four random bits anywhere in the word,
adjacent bursts including the double bit model of
NVBitFI~\cite{tsai2021nvbitfi}, two exponent bits, stuck at lines and whole
word replacement under instruction level injection. Faults above the rounding
bound are recovered at $0.682$ to $1.000$ with BF16 operands and $0.690$ to
$1.000$ with FP16, and no pattern produced a false positive. Recovery rises
with the number of random bits flipped, from $0.423$ for one bit anywhere to
$0.936$ for four, because more flipped bits yield larger errors. Adjacent
bursts are the hardest pattern, at $0.440$ and $0.380$, since a burst confined
to low mantissa bits often stays below the bound. Two bit exponent upsets are
recovered at $0.987$ in both formats. Stuck at 1 lines are almost always silent.
One BF16 trial in $40$ produced an observable, finite error, which was
recovered, and no FP16 trial produced one. For values of normal magnitude, the
exponent bits such a line pins are already set, which is the data dependence
that lets stuck at faults survive in the field.

\subsection{RQ3: Real Models}
\label{sec:rq3}

\begin{table}[t]
\centering
\caption{Perplexity on wikitext-2 ($4096$ tokens), $8$ output side faults on bit
         26 per guarded GEMM per call. Guarded: all $32$ Llama-2-7B
         \texttt{down\_proj} and all $36$ GPT-2 large \texttt{c\_proj} layers,
         FP32 delivery. Each model is hosted entirely in the stated operand
         format.}
\label{tab:realmodel}
\footnotesize
\setlength{\tabcolsep}{2.5pt}
\begin{tabular}{@{}l l r r r r r@{}}
\toprule
Model & Op. & Clean & Unguarded & Guarded & Recovered & FP \\
\midrule
Llama-2-7B  & BF16 & 7.4629  & 7.6965  & 7.4614  & 2047/2048 & 0 \\
Llama-2-7B  & FP16 & 7.4650  & 7.6132  & 7.4651  & 2048/2048 & 0 \\
GPT-2 large & BF16 & 26.2301 & 26.7666 & 26.2301 & 2303/2304 & 0 \\
GPT-2 large & FP16 & 26.2223 & 26.6885 & 26.2223 & 2304/2304 & 0 \\
\bottomrule
\end{tabular}
\end{table}

\begin{table}[t]
\centering
\caption{Recovery under NVBit accumulator injection on a real Llama-2-7B
         \texttt{down\_proj} layer ($4096{\times}11008{\times}4096$, activations
         from a wikitext-2 forward pass) and on synthetic operands of the same
         shape, calibrated engine at $\rho_{\min}=0.02$, two sites per trial,
         $40$ trials per cell, FPR $0$ in every cell. Products delivered correct
         after repair and the second probe: $157$ of $160$ on the real operands
         and $78$ of $80$ on synthetic ones.}
\label{tab:llamanvbit}
\footnotesize
\setlength{\tabcolsep}{3pt}
\begin{tabular}{@{}l c c@{}}
\toprule
Operands, fault & BF16 & FP16 \\
\midrule
Llama, 1 bit            & 0.988 & 0.988 \\
Llama, 2 exponent bits  & 1.000 & 0.975 \\
Synthetic, 1 bit        & 0.988 & 0.988 \\
\bottomrule
\end{tabular}
\end{table}

\begin{table}[t]
\centering
\caption{Training damage versus fault count, GPT-2 \texttt{h.5.mlp.c\_proj},
         BF16 operands, FP32 delivery, output side faults on bit 26, calibrated
         engine with a declared budget of $s = 16$. Damage is
         $\max|\text{loss} - \text{loss}_{\text{clean}}|$ after injection.}
\label{tab:dose}
\small
\begin{tabular}{@{}r r r r@{}}
\toprule
Faults & Repaired & Unguarded & Guarded \\
\midrule
16   & 16/16    & $1.93\times10^{-6}$ & $5.45\times10^{-7}$ \\
64   & 64/64    & $1.24\times10^{-5}$ & $1.09\times10^{-6}$ \\
256  & 256/256  & $1.45\times10^{-4}$ & $8.94\times10^{-7}$ \\
1024 & 379/1024 & $7.62\times10^{-2}$ & $7.44\times10^{-5}$ \\
4096 & 381/4096 & $3.50\times10^{-1}$ & $6.69\times10^{-4}$ \\
\bottomrule
\end{tabular}
\end{table}

Guarding every MLP down projection of Llama-2-7B (Table~\ref{tab:realmodel})
recovers $2047$ of $2048$ faults with BF16 operands and all $2048$ with FP16,
and every guarded call is flagged. The guard removes $99.4\%$ and $99.9\%$ of
the perplexity damage. GPT-2 large recovers $2303$ of $2304$ and $2304$ of
$2304$, and its guarded perplexity matches the clean one to four digits in both
formats. No run produced a false positive. With every one of the $256$ Llama
calls and $288$ GPT-2 large calls corrupted and localized, a guarded pass took
$48.2$\,s (BF16) and $41.5$\,s (FP16) for Llama-2-7B against $3.2$ to $3.5$\,s
clean, and $75.9$\,s and $25.8$\,s for GPT-2 large against $1.2$ to $1.8$\,s.

Table~\ref{tab:dose} shows a dose threshold in training. Sixteen faults leave
the loss essentially unchanged, and the damage becomes large by $1024$.
Repair is complete up to $256$ faults, which is $16\times$ the declared budget.
Beyond that, localization keeps at most $K = 128$ candidates in each of its
three rounds and repairs $379$ and $381$ faults. Even so, the guard lowers the
damage about $1000\times$ at $1024$ faults and $520\times$ at $4096$. Both
experiments use output side injection and, by \S\ref{sec:rq2}, bound recovery
from above.

Table~\ref{tab:llamanvbit} repeats the measurement under NVBit on the operands
of a real Llama-2-7B \texttt{down\_proj} layer, whose inner dimension of $11008$
raises sketch noise. At the default $\rho_{\min} = 0.02$, with no manual
setting, the calibrated engine recovers $0.975$ to $1.000$ on the real operands
and $0.988$ on synthetic operands of the same shape. For these products of the
tiled kernel, the measured noise sets $m_{\mathrm{loc}} = 1024$. The probe does
not slow down as $m$ adapts. Measured with the GPU to itself, it costs $0.62$ to
$0.69$\,ms on cuBLAS products and $0.93$ to $0.94$\,ms at the recalibrated
$m = 726$ on products of the tiled kernel, beside a GEMM of $0.46$ to
$0.47$\,ms, because reading $A$ and $C$ dominates the probe at this shape.

\subsection{RQ4: Cost}
\label{sec:rq4}

\begin{figure}[t]
\centering
\begin{tikzpicture}
\begin{axis}[
  width=\columnwidth, height=4.3cm,
  ybar=0pt, bar width=7pt,
  ymode=log, log origin=infty, ymin=0.1, ymax=20,
  ylabel={Milliseconds (log)},
  symbolic x coords={qkv,ffnup,attn,square},
  xtick=data,
  xticklabels={qkv,ffn up,attn 70B,square},
  xticklabel style={font=\footnotesize},
  tick label style={font=\footnotesize}, label style={font=\footnotesize},
  axis line style={inkMuted}, tick style={inkMuted},
  ymajorgrids, grid style={gridLine},
  enlarge x limits=0.18,
  legend style={font=\footnotesize, draw=none, fill=none},
  legend columns=2, legend pos=north west]
\addplot[fill=seriesB, draw=white, line width=1pt] coordinates
  {(qkv,0.346) (ffnup,1.200) (attn,2.844) (square,12.474)};
\addplot[fill=seriesA, draw=white, line width=1pt] coordinates
  {(qkv,0.777) (ffnup,0.993) (attn,1.240) (square,3.062)};
\legend{BF16 GEMM, probe}
\end{axis}
\end{tikzpicture}
\caption{Deployable probe cost against the guarded GEMM, BF16 operands, FP32
delivery, with activations rotating so the sketch of the activation side cannot
be reused. Shapes: qkv $8192{\times}4096{\times}4096$, ffn up
$8192{\times}4096{\times}14336$, attn 70B $16384{\times}8192{\times}8192$,
square $16384^3$. FP16: probe $0.780$ to $6.675$\,ms against GEMMs of
$0.360$ to $12.736$\,ms.}
\label{fig:cost}
\end{figure}
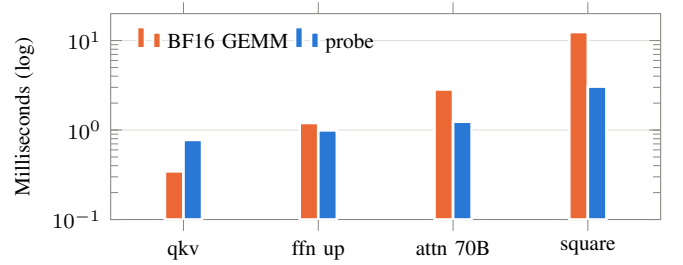

\emph{Clean path.} The $S$ only probe (\S\ref{sec:probe}) is the only stage paid
on every call. Figure~\ref{fig:cost} measures it with activations rotating,
which is the only setting that corresponds to a forward pass. The probe takes
$0.777$ to $3.062$\,ms against GEMMs of $0.346$ to $12.474$\,ms, that is,
$25$ to $225\%$ of the GEMM and $8.1$ to $21.0\times$ the bandwidth floor
$(\pi/\beta)(1/n_3 + 2/n_2)$ for reading $A$ once and $C$ once. The term
$1/n_3 + 2/n_2$ governs the ratio, so shallow projections are expensive and a
decision per layer follows directly. Segmented scatters and replay from a CUDA
graph with a deferred verdict account for most of the reduction from the
unoptimized path. The remaining gap to the floor is launch and synchronization
overhead rather than bandwidth. With FP16 operands, recalibration moved the
probe at $16384{\times}8192{\times}8192$ from the law's $m = 1142$ to $m = 656$,
where it costs $1.259$\,ms. The FP16 square shape ran at the law's $m = 2397$.

\emph{Dirty path.} Localization time follows the search radius, as
\S\ref{sec:radius} describes. At the Llama-2-7B \texttt{down\_proj} shape, the
calibrated engine localizes in $8.5$ to $9.0$\,ms on cuBLAS products
($m_{\mathrm{loc}} = 448$ to $464$, $r = 4$) and in $13.6$\,ms on products of
the tiled kernel ($m_{\mathrm{loc}} = 1024$, $r = 15$), against a GEMM of $0.46$
to $0.47$\,ms. The floor $\tau_d$ accounts for the second case. At almost the
same radius, localizing those products took $1162$ to $1176$\,ms before it was
added, because clean buckets that cleared $\tau_c$ walked the whole search. At
$16384{\times}8192{\times}8192$, localization takes $42$ to $43$\,ms
($m_{\mathrm{loc}} = 2896$, $r = 8$ to $9$) against a GEMM of $2.86$ to
$2.96$\,ms, down from $1179$ to $1191$\,ms at $r = 16$. A deployment that wants
only a correct product should still recompute on detection. Localization earns
its place through telemetry (\S\ref{sec:rq6}), not through repair economics.

\subsection{RQ5: Comparison with Baselines}
\label{sec:rq5}

\begin{table}[t]
\centering
\caption{Baselines on identical BF16 operands and identical faults, with the
         same analytic threshold form. Det.: fraction of trials detected.
         $s{=}1$: recovery with one output side fault at qkv, $20$ trials.
         $s{=}2$: recovery with two NVBit accumulator faults at
         $4096{\times}2048{\times}4096$, $40$ trials. FP: false positives over
         both workloads. ms: median cost per call at qkv to detect and
         localize, measured while NVBit campaigns shared the GPU.}
\label{tab:baselines}
\footnotesize
\setlength{\tabcolsep}{3pt}
\begin{tabular}{@{}l c c c c c@{}}
\toprule
Method & Det. & $s{=}1$ & $s{=}2$ & FP & ms \\
\midrule
\sys{}            & 1.000 & 1.000 & 1.000 & 0   & 9.884 \\
Checksum ABFT     & 1.000 & 1.000 & 0.025 & 0   & 0.462 \\
Weighted checksum & 1.000 & 1.000 & 0.000 & 2   & 0.697 \\
Freivalds         & 1.000 & n/a   & n/a   & n/a & 0.253 \\
Recompute         & 1.000 & 1.000 & 0.975 & 0   & 5.824 \\
\bottomrule
\end{tabular}
\end{table}

Table~\ref{tab:baselines} compares \sys{} with checksum ABFT, weighted checksum
ABFT, a Freivalds probe, and recompute and compare on the same half precision
operands and the same corrupted products. All five detect every corrupted
product, and all except Freivalds localize a single fault. Two faults per
product separate them. Checksum ABFT localizes $2$ of $80$, weighted checksum
ABFT localizes none and reports $2$ false positives, recompute and compare
localizes $78$, and \sys{} localizes all $80$. On a corrupted call at qkv,
\sys{} took $9.9$\,ms to detect and localize, against $5.8$\,ms for a
recompute measured under the same shared conditions. Row and column checksums flag two rows
and two columns for two errors, which leaves four candidate positions, and a
weighted syndrome averages two indices. Both therefore localize only single
errors by construction. FT-GEMM is not measured. It is an FP32$\times$FP32
kernel of its own rather than a wrapper around the vendor GEMM, so it cannot
run this workload, and its overhead would be relative to a GEMM roughly an
order of magnitude slower than the tensor core path. Table~\ref{tab:related}
compares it qualitatively.

\subsection{RQ6: What Localization Adds}
\label{sec:rq6}

A recompute yields a correct product and a checksum yields one bit. Neither
tells an operator which device to pull. Each localized fault carries its
coordinate, its magnitude, and its direction, and aggregated records support
three diagnoses that need progressively more data. The first is a
\emph{repeated coordinate} across independent products. Under a transient
model, a coincidence among $F$ faults on $n_1n_3$ entries has probability
$\approx 1-e^{-F(F-1)/2n_1n_3}$, about $1.5\times10^{-4}$ for $F = 18$ on an
output of $2^{20}$ entries. The second is \emph{concentration within a tile}
of $(i \bmod t_1, j \bmod t_3)$, which needs roughly $640$ faults for a $\chi^2$
test over a $16\times8$ tile. The third is a device whose \emph{rate} is an
outlier against the fleet.

The flipped bit position is \emph{not} recoverable, and the reason is
structural. Inferring it from $C_{ij}$ and $\gamma$ identifies the bit in
$400/400$ faults written into a finished product and in $0/80$ under NVBit. An
upset in the middle of accumulation is followed by the remaining partial
products, so the delivered entry is not the clean value with one bit inverted.
We therefore report coordinate, magnitude and direction only. Because every
candidate is validated, a degraded localizer yields \emph{incomplete} telemetry
rather than wrong telemetry, which is another reason to verify at accumulation
precision.

\section{Discussion and Limitations}
\label{sec:discussion}

\emph{Trust.} The guarantee against false positives assumes a recomputation that
the fault cannot reach (\S\ref{sec:threat}), so a permanent fault on the
verifying device is not covered. \emph{Scaling.} Equation~\eqref{eq:mnum} makes
$m$ grow as $n^{2.57}$ for square products. At $65536^3$ the law asks for
$25.8$\,GB of sketches, and the population cap stops $m$ at $16384$, which
bounds the method before it bounds accuracy. Transformer projections are
rectangular and far below that size (Table~\ref{tab:rect}).
\emph{Portability.} The constants in~\eqref{eq:sigma} are specific to the BLAS,
the device and the operand format, and \sys{} measures the noise at run time
rather than relying on them (\S\ref{sec:radius}). \emph{Economics.} The probe is
the steady state cost and is expensive on shallow projections. Localization
costs $15$ to $30$ GEMMs at the shapes measured, so repair still costs more than
a recompute and localization is justified by the diagnostic record.
\emph{Coverage.} Operand corruption, faults outside matrix products, and fused
kernels that never materialize their products are out of reach of any verifier
that runs after the fact. \emph{Evaluation.} Real model perplexity uses output
side injection. Instruction level injection on the operands of a real
Llama-2-7B layer recovers $0.988$ of single bit faults with the calibrated
engine, the same as on synthetic operands of that shape.

\section{Conclusion}
\label{sec:conclusion}

\sys{} verifies an unmodified mixed precision GEMM after the fact, at the
precision of its accumulator, and localizes several corrupted entries with no
false positives by construction. Making hashed moment localization work in
floating point required a measured noise law whose constant depends on the
library, and deploying it required measuring that noise at run time. Evaluating
it honestly required instruction level injection, because output side injection
cannot produce the small faults that a live accumulator yields and so overstates
what a localizer recovers. Localization is not the cheapest way to obtain a
correct product, but of the three responses it is the only one that says where
the corruption was.

\bibliographystyle{IEEEtran}
\bibliography{references}

\end{document}